\documentclass[3p,times]{elsarticle}

\usepackage{ecrc}

\volume{00}

\firstpage{1}

\journalname{}

\runauth{Chen Wang et al.}

\jid{procs}

\jnltitlelogo{}

\CopyrightLine{2023}{Published by Elsevier Ltd.}

\usepackage{amssymb}
\usepackage[figuresright]{rotating}

\usepackage{amsmath}
\usepackage{amsthm}
\usepackage{graphics}
\usepackage{subfigure}
\usepackage{lineno,hyperref}
\usepackage{cases}
\usepackage{bm}
\usepackage[linesnumbered,ruled,vlined]{algorithm2e}
\usepackage{hyperref}
\usepackage[noend]{algpseudocode}
\modulolinenumbers[1]

\newtheorem{proposition}{Proposition}

\newtheoremstyle{prestyle}
{0} % Space above
{\topsep} % Space below
{\itshape} % Body font
{} % Indent amount
{\bfseries} % Theorem head font
{.} % Punctuation after theorem head
{.5em} % Space after theorem head
{} % Theorem head spec (can be left empty, meaning `normal')
\theoremstyle{prestyle}

\theoremstyle{definition}	

\begin{document}

\begin{frontmatter}



\dochead{}
%% e.g. \dochead{17th International Conference on Dynamical Processes in Excited States of Solids}

\title{Efficient Calculations for Inverse of $k$-diagonal Circulant Matrices and Cyclic Banded Matrices}


\author[]{Chen Wang}
\ead{2120220677@mail.nankai.edu.cn}
\author[]{Hailong Yu}
\ead{hlyu@mail.nankai.edu.cn}
\author[]{Chao Wang\corref{cor1}}
\ead{wangchao@nankai.edu.cn}

\address[mymainaddress]{Address: College of Software, Nankai University, Tianjin 300371, China}

\cortext[cor1]{Corresponding author.}

\begin{abstract}
	$k$-diagonal circulant matrices and cyclic banded matrices are widely used in numerical simulations and signal processing of circular linear systems. Algorithms that directly involve or specify linear or quadratic complexity for the inverses of these two types of matrices are rare. We find that the inverse of a $k$-diagonal circulant matrix can be uniquely determined by a recursive formula, which can be derived within $O(k^3 \log n+k^4)$. Similarly for the inverse of a cyclic banded matrix, its inverse can be uniquely determined by a series of recursive formulas, with the initial terms of these recursions computable within $O(k^3 n+k^5)$. The additional costs for solving the complete inverses of these two types of matrices are $kn$ and $kn^2$. Our calculations enable rapid representation with most processes defined by explicit formulas. Additionally, most algorithms for inverting $k$-diagonal circulant matrices rely on the Fast Fourier Transform, which is not applicable to finite fields, while our algorithms can be applied to computations in finite fields.
	
	%Calculating the inverse of $k$-diagonal circulant matrices and cyclic banded matrices is a more challenging problem than calculating their determinants. Algorithms that directly involve or specify linear or quadratic complexity for the inverses of these two types of matrices are rare. This paper presents two fast algorithms that can compute the complexity of a $k$-diagonal circulant matrix within complexity $O(k^3 \log n+k^4)+kn$, and for $k$-diagonal cyclic banded matrices it is $O(k^3 n+k^5)+kn^2$. Since $k$ is generally much smaller than $n$, the cost of these two algorithms can be approximated as $kn$ and $kn^2$.
	
\end{abstract}


\end{frontmatter}

%\linenumbers

\section{Introduction and preliminaries \label{s1}}

$k$-diagonal circulant matrices and cyclic banded matrices are two special types of matrices. Due to their characteristic of periodic boundaries, they are widely used in various linear systems such as digital signal processing, image compression, number theory, automata theory, coding theory, cryptography, and more \cite{ahmed2012orthogonal,duhamel1987improved,lu1989algorithms,del2015note,vaidyanathan1999cyclic}.

A $k$-diagonal circulant matrix ($k$-CM) is a type of sparse circulant matrix with only $k$ non-zero diagonals, where each row of the matrix is a cyclic shift of one position relative to the previous row. A $k$-diagonal cyclic banded matrix ($k$-CBM) has its non-zero elements primarily concentrated around the main diagonal and a few diagonals close to it, forming a banded structure. Unlike the conventional banded matrix, the banded structure of $k$-CBM is cyclic at the ends of the matrix, meaning the top and bottom corners are also considered adjacent, forming a ring-like band. The structure of these two matrices is as shown in Eq. \ref{M}. If for each $i,j,p$, $x_{i}^{(j)}=x_{i}^{(p)}$, then the matrix $M$ is a $k$-CM; conversely, if exists $x_{i}^{(j)}\neq x_{i}^{(p)}$, then the matrix is a $k$-CBM. Besides, for each $i$, we need to ensure that $x_{k}^{(i)} \neq 0$

\begin{equation}
\label{M}
\setlength{\arraycolsep}{3.5pt}
M = 
\begin{bmatrix}
	x_{3}^{(1)} & x_{4}^{(1)} & x_{5}^{(1)} & 0 & \cdots & 0 & 0 & x_{1}^{(1)} & x_{2}^{(1)}\\
	x_{2}^{(2)} & x_{3}^{(2)} & x_{4}^{(2)} & x_{5}^{(2)} & 0 & \cdots & 0 & 0 & x_{1}^{(2)}\\
	x_{1}^{(3)} & x_{2}^{(3)} & x_{3}^{(3)} & x_{4}^{(3)} & x_{5}^{(3)} & 0 & \cdots & 0 & 0\\
	0 & x_{1}^{(4)} & x_{2}^{(4)} & x_{3}^{(4)} & x_{4}^{(4)} & x_{5}^{(4)} & 0 & \cdots & 0\\
	\vdots & \ddots & \ddots & \ddots & \ddots & \ddots & \ddots & \ddots & \vdots\\
	0 & \vdots & 0 & x_{1}^{(n-3)} & x_{2}^{(n-3)} & x_{3}^{(n-3)} & x_{4}^{(n-3)} & x_{5}^{(n-3)} & 0\\
	0 & 0 & \vdots & 0 & x_{1}^{(n-2)} & x_{2}^{(n-2)} & x_{3}^{(n-2)} & x_{4}^{(n-2)} & x_{5}^{(n-2)}\\
	x_{5}^{(n-1)} & 0 & 0 & \vdots & 0 & x_{1}^{(n-1)} & x_{2}^{(n-1)} & x_{3}^{(n-1)} & x_{4}^{(n-1)}\\
	x_{4}^{(n)} & x_{5}^{(n)} & 0 & 0 & \cdots & 0 & x_{1}^{(n)} & x_{2}^{(n)} & x_{3}^{(n)}
\end{bmatrix}_{n*n}.
\end{equation}

Existing research primarily focuses on the determinant of $k$-CM. Now, the complexity of calculating the determinant of $k$-CM has been reduced to linear \cite{jia2017efficient,jia2016numerical} or even logarithmic $O(k\log k \log n+k^4)$ \cite{CHEN2014201}. 

For the inverse of $k$-CM, traditional algorithms capable of solving $k$-CBM include the Fast Fourier Transform (FFT), with a complexity of $O(n \log n)$  \cite{davis1979circulant,gray2006toeplitz}. However, FFT has several disadvantages: firstly, it cannot utilize the sparse characteristic of $k$-diagonal; secondly, its complexity is currently difficult to satisfy; thirdly, it fails when the problem comes to a finite field; and finally, it cannot be used for calculating $k$-CBM. S. Serra-Capizzano et al. introduced an efficient multigrid method that employs iterators to compute the inverse of banded circulant matrices, achieving linear overall complexity \cite{serra2004multigrid}. 
For the inverse of $k$-CBM, the primary methods currently used are Gaussian elimination, both of which have a complexity of $O(kn^2)$. P. R\'{o}zsa et al. have thoroughly analyzed generalized band matrices and periodic block-tridiagonal matrices. Their algorithm for solving the inverse of $k$-CBM has a complexity of $O(k^2n^2)$ \cite{favati1991generalized,rozsa1992periodic}. 

For $k$-CM and $k$-CBM, current algorithms applicable to Toeplitz matrices do not extend well. For example, the Gohberg-Semencul (GS) formula, which computes the inverse of Toeplitz matrices, $O(n\log ^2n)$ \cite{bini2012polynomial,heinig1984algebraic}. More efficient algorithms exist for $k$-CM due to its sparse and circulant properties, and for $k$-CBM, which lacks a Toeplitz structure, these methods are not applicable. This leads to a scarcity of algorithms for computing their inverses.

The superfast method allows for the representation of the inverse without explicitly computing all its entries, reducing complexity to sub-linear levels, such as the GS formula mentioned earlier for solving the inverse of Toeplitz matrices since the inverse can be identified from a small set of parameters. For $k$-CM and $k$-CBM, rapid representations of the inverses also exist, and we have summarized the calculation process into formulas. We find that the inverse of a $k$-CM can be uniquely represented by a specific recursive formula, which can be derived within $O(k^3 \log n+k^4)$. This makes the computation of its inverse extremely fast. If it is necessary to explicitly compute all $n$ entries of the inverse, the additional cost is $kn$. Similarly for the inverse of $k$-CBM, its inverse can be uniquely determined by a series of recursive formulas. The initial terms of these recursions can be computed within $O(k^3 n+k^5)$, and the additional cost is $kn^2$. Tables \ref{table1} and \ref{table2} provide a comparison of the computational efficiencies for these two types of matrices. It can be observed that our calculations not only allow for rapid representation but also significantly reduce the coefficient of $n$. This makes our computations highly efficient when calculating the inverses of large-scale $k$-CM and $k$-CBM.

\begin{table}[h]
\center
\caption{Comparison of inverse of $k$-CM}
\label{table1}
\begin{tabular}{lccc}
	\hline
	& rapid representation    & overall complexity & applicability over finite field \\
	\hline
	our computation    & $O(k^3 \log n+k^4)$ & $O(k^3 \log n+k^4)+kn$  & yes  \\
	Fast Fourier Transform   & - & $O(n\log n)$ & no \\
	multigride method \cite{serra2004multigrid}  & - & $>2kn$  & no  \\
	\hline
\end{tabular}
\end{table}

\begin{table}[h]
\center
\caption{Comparison of inverse of $k$-CBM}
\label{table2}
\begin{tabular}{lccc}
	\hline
	& rapid representation    & overall complexity\\
	\hline
	our computation    & $O(k^3 n+k^5)$ & $O(k^3 n+k^5)+kn^2$ \\
	Gaussian elimination   & - & $O(kn^2)$\\
	\hline
\end{tabular}
\end{table}

This paper consists of five sections. The second section introduces the method for calculating the determinant of $k$-CM, which is essential for inversion. The third section presents the fast algorithms for inverting $k$-CM while the fourth section is for $k$-CBM. The final section summarizes the work of the entire paper.

\section{Determinants of $k$-diagonal circulant matrices \label{s2}}

In this section, we take $k=5$ and $n > 8$ as an example and provide a general formula applicable to all $k$ at the end. This is the necessary preparation for calculating the inverse of $k$-CM.

Firstly, we can transform a $k$-CM $M$ into $M'$ as in Eq.\ref{M'} by swapping its columns. We have $\begin{vmatrix}M\end{vmatrix} = (-1)^{2n-4} \begin{vmatrix}M'\end{vmatrix} = \begin{vmatrix}M'\end{vmatrix}$. In subsequent calculations, we will no longer categorize based on the swapped versions of the matrix.
\begin{equation}
\setcounter{MaxMatrixCols}{20}
\label{M'}
M' = 
\begin{bmatrix}
	x_{5} & 0 & 0 & \cdots & 0 & x_{1} & x_{2} & x_{3} & x_{4}\\
	x_{4} & x_{5} & 0 & 0 & \cdots & 0 & x_{1} & x_{2} & x_{3}\\
	x_{3} & x_{4} & x_{5} & 0 & 0 & \cdots & 0 & x_{1} & x_{2}\\
	x_{2} & x_{3} & x_{4} & x_{5} & 0 & 0 & \cdots & 0 & x_{1}\\
	x_{1} & x_{2} & x_{3} & x_{4} & x_{5} & 0 & 0 & \ddots & 0\\
	\vdots & \ddots & \ddots & \ddots & \ddots & \ddots & \ddots & \ddots & \vdots\\
	0 & \cdots & 0 & x_{1} & x_{2} & x_{3} & x_{4} & x_{5} & 0\\
	0 & 0 & \cdots & 0 & x_{1} & x_{2} & x_{3} & x_{4} & x_{5}
\end{bmatrix}_{n*n}.
\end{equation}
Next, we use the first column of the matrix $M'$ to set other elements of the first row of $M'$ to 0, as shown locally in Eq.\ref{xy}.
\begin{equation}
\label{xy}
\renewcommand\arraystretch{1.3}
\begin{bmatrix}
	x_{5} & \cdots & y_{11} & y_{12} & y_{13} & y_{14}\\
	x_{4} & \cdots & y_{21} & y_{22} & y_{23} & y_{24}\\
	x_{3} & \cdots & y_{31} & y_{32} & y_{33} & y_{34}\\
	x_{2} & \cdots & y_{41} & y_{42} & y_{43} & y_{44}\\
	x_{1} & \cdots & 0 & 0 & 0 & 0\\
	\vdots & \ddots  & \vdots & \vdots & \vdots & \vdots
\end{bmatrix}
\Longrightarrow 
\begin{bmatrix}
	x_{5} & \cdots & 0 & 0 & 0 & 0\\
	x_{4} & \cdots & y_{21}-\frac{y_{11}x_{4}}{x_{5}} & y_{22}-\frac{y_{12}x_{4}}{x_{5}} & y_{23}-\frac{y_{13}x_{4}}{x_{5}} & y_{24}-\frac{y_{14}x_{4}}{x_{5}}\\
	x_{3} & \cdots & y_{31}-\frac{y_{11}x_{3}}{x_{5}} & y_{32}-\frac{y_{12}x_{3}}{x_{5}} & y_{33}-\frac{y_{13}x_{3}}{x_{5}} & y_{34}-\frac{y_{14}x_{3}}{x_{5}}\\
	x_{2} & \cdots & y_{41}-\frac{y_{11}x_{2}}{x_{5}} & y_{42}-\frac{y_{12}x_{2}}{x_{5}} & y_{43}-\frac{y_{13}x_{2}}{x_{5}} & y_{44}-\frac{y_{14}x_{2}}{x_{5}}\\
	x_{1} & \cdots & -\frac{y_{11}x_{1}}{x_{5}} & -\frac{y_{12}x_{1}}{x_{5}} & -\frac{y_{13}x_{1}}{x_{5}} & -\frac{y_{14}x_{1}}{x_{5}}\\
	\vdots & \ddots  & \vdots & \vdots & \vdots & \vdots
\end{bmatrix}.
\end{equation}
Meanwhile, another equation is very important, as shown in Eq.\ref{xty}.
\begin{equation}
\label{xty}
\renewcommand\arraystretch{1.3}
\begin{bmatrix}
	y_{21}-\frac{y_{11}x_{4}}{x_{5}} & y_{22}-\frac{y_{12}x_{4}}{x_{5}} & y_{23}-\frac{y_{13}x_{4}}{x_{5}} & y_{24}-\frac{y_{14}x_{4}}{x_{5}}\\
	y_{31}-\frac{y_{11}x_{3}}{x_{5}} & y_{32}-\frac{y_{12}x_{3}}{x_{5}} & y_{33}-\frac{y_{13}x_{3}}{x_{5}} & y_{34}-\frac{y_{14}x_{3}}{x_{5}}\\
	y_{41}-\frac{y_{11}x_{2}}{x_{5}} & y_{42}-\frac{y_{12}x_{2}}{x_{5}} & y_{43}-\frac{y_{13}x_{2}}{x_{5}} & y_{44}-\frac{y_{14}x_{2}}{x_{5}}\\
	-\frac{y_{11}x_{1}}{x_{5}} & -\frac{y_{12}x_{1}}{x_{5}} & -\frac{y_{13}x_{1}}{x_{5}} & -\frac{y_{14}x_{1}}{x_{5}}\\
\end{bmatrix}
=
\begin{bmatrix}
	-\frac{x_{4}}{x_{5}} & 1 & 0 & 0 \\
	-\frac{x_{3}}{x_{5}} & 0 & 1 & 0 \\
	-\frac{x_{2}}{x_{5}} & 0 & 0 & 1 \\
	-\frac{x_{1}}{x_{5}} & 0 & 0 & 0
\end{bmatrix}
\times 
\begin{bmatrix}
	y_{11} & y_{12} & y_{13} & y_{14}\\
	y_{21} & y_{22} & y_{23} & y_{24}\\
	y_{31} & y_{32} & y_{33} & y_{34}\\
	y_{41} & y_{42} & y_{43} & y_{44}
\end{bmatrix}.
\end{equation}
The matrix after one transformation is shown as in Eq. \ref{B1}.
\begin{equation}
\label{B1}
\begin{bmatrix}
	x_{5} & 0 & \cdots & 0 & 0 & 0 & 0 & 0 & 0 & 0 & 0\\
	x_{4} & x_{5} & 0 & \cdots & 0 & 0 & 0 & \\
	x_{3} & x_{4} & x_{5} & 0 & \cdots & 0 & 0 \\
	x_{2} & x_{3} & x_{4} & x_{5} & 0 & \cdots & 0 &  & B_{1}\\
	x_{1} & x_{2} & x_{3} & x_{4} & x_{5} & 0 & \ddots & \\
	0 & x_{1} & x_{2} & x_{3} & x_{4} & x_{5} & 0 & \ddots & 0 & 0 & 0\\
	\vdots & \ddots & \ddots & \ddots & \ddots & \ddots & \ddots & \ddots & \vdots & \vdots\\
	0 & \cdots & 0 & 0 & 0 & x_{1} & x_{2} & x_{3} & x_{4} & x_{5} & 0\\
	0 & 0 & \cdots & 0 & 0 & 0 & x_{1} & x_{2} & x_{3} & x_{4} & x_{5}
\end{bmatrix}_{n*n}.
\end{equation}
In this way, we can extract an $x_{5}$, and we have

\begin{equation}
B_{1} = TB = 
\begin{bmatrix}
	-\frac{x_{4}}{x_{5}} & 1 & 0 & 0 \\
	-\frac{x_{3}}{x_{5}} & 0 & 1 & 0 \\
	-\frac{x_{2}}{x_{5}} & 0 & 0 & 1 \\
	-\frac{x_{1}}{x_{5}} & 0 & 0 & 0
\end{bmatrix}\times
\begin{bmatrix}
	x_{1} & x_{2} & x_{3} & x_{4}\\
	0 & x_{1} & x_{2} & x_{3}\\
	0 & 0 & x_{1} & x_{2}\\
	0 & 0 & 0 & x_{1}
\end{bmatrix}.
\end{equation}

So, each time we perform the operation in Eq.\ref{xy}, we can reduce the size of the matrix and extract a $x_{5}$. Eventually, the matrix will be reduced to an $8 \times 8$ matrix as shown in Eq.\ref{24}.

\begin{equation}
\label{24}
\begin{bmatrix}
	x_{5} & 0 & 0 & 0 &  &  &  & \\
	x_{4} & x_{5} & 0 & 0 &  &  &  & \\
	x_{3} & x_{4} & x_{5} & 0 &  & B_{n-8} &  & \\
	x_{2} & x_{3} & x_{4} & x_{5} &  &  &  & \\
	x_{1} & x_{2} & x_{3} & x_{4} & x_{5} & 0 & 0 & 0\\
	0 & x_{1} & x_{2} & x_{3} & x_{4} & x_{5} & 0 & 0\\
	0 & 0 & x_{1} & x_{2} & x_{3} & x_{4} & x_{5} & 0\\
	0 & 0 & 0 & x_{1} & x_{2} & x_{3} & x_{4} & x_{5}
\end{bmatrix}.
\end{equation}
According to Eq.\ref{xty}, we can get the matrix $B'$.
\begin{equation}
\label{25}
B_{n-8} = T^{n-8}B = 
\begin{bmatrix}
	-\frac{x_{4}}{x_{5}} & 1 & 0 & 0 \\
	-\frac{x_{3}}{x_{5}} & 0 & 1 & 0 \\
	-\frac{x_{2}}{x_{5}} & 0 & 0 & 1 \\
	-\frac{x_{1}}{x_{5}} & 0 & 0 & 0
\end{bmatrix}^{n-8}\times
\begin{bmatrix}
	x_{1} & x_{2} & x_{3} & x_{4}\\
	0 & x_{1} & x_{2} & x_{3}\\
	0 & 0 & x_{1} & x_{2}\\
	0 & 0 & 0 & x_{1}
\end{bmatrix}.
\end{equation}
Now we get the answer shown in Eq.\ref{M} when $k=5$.
\begin{equation}
\label{df}
\begin{vmatrix}M\end{vmatrix} = (-1)^{(k-i)(n-k+i)} x_{k}^{n-k+1} \begin{vmatrix}D-CA^{-1}T^{n-2k+2}B\end{vmatrix},
\end{equation}
where the element in the upper left corner of the matrix is $x_{i}$ ($i\leq k$, in the example of Eq.\ref{M}, $i = 3$) and 
\begin{equation}
\label{abcdt}
A = D =
\begin{bmatrix}
	x_{k} & 0 & \cdots & 0\\
	x_{k-1} & x_{k} & \cdots & 0\\
	\vdots & \vdots & \ddots & \vdots\\
	x_{2} & x_{3} & \cdots & x_{k}
\end{bmatrix}
,B = C = 
\begin{bmatrix}
	x_{1} & x_{2} & \cdots & x_{k-1}\\
	0 & x_{1} & \cdots & x_{k-2}\\
	\vdots & \vdots & \ddots & \vdots\\
	0 & 0 & \cdots & x_{1}
\end{bmatrix},
T = 
\begin{bmatrix}
	-\frac{x_{k-1}}{x_{k}} & 1 & 0  & \cdots & 0\\
	-\frac{x_{k-2}}{x_{k}} & 0 & 1  & \cdots & 0\\
	\vdots & \vdots & \ddots & \ddots & \vdots\\
	-\frac{x_{2}}{x_{k}} & 0 & \cdots & 0 & 1\\
	-\frac{x_{1}}{x_{k}}  & 0 & \cdots & 0 & 0
\end{bmatrix}.
\end{equation}

\begin{proposition}
We can solve the determinant of a $k$-CM in $O(k^3 \log n)$.
\end{proposition}
\begin{proof}
Since matrices $A$, $B$, $C$, and $D$ are all of order $k-1$, the complexity of their computation is $O{(k^3)}$. For $T^{n-2k+2}$, we can use the fast matrix exponentiation method to calculate it.
\begin{equation}	
	T^n=
	\begin{cases}
		T^{\frac{n}{2}}*T^{\frac{n}{2}} & \text{ if } n \text{ mod } 2 =0 \\
		T^{\frac{n-1}{2}} * T^{\frac{n-1}{2}}* T & \text{ if } n \text{ mod } 2 =1 
	\end{cases}
\end{equation}

By doing so, we reduce the computation of the $n$th power of $T$ to logarithmic complexity, hence the complexity is $O(k^3 \log n)$. Furthermore, using the Jordan canonical form method can also solve the problem with logarithmic complexity, but when the problem is approached within a finite field, errors caused by inaccuracies may occur.
\end{proof}

The algorithm can solve the determinant of a $k$-CM in $O(k^3 \log n)$. Similar computation will also be utilized in the next section.

\section{Inverse of $k$-diagonal circulant matrices\label{s3}}

The algorithm mentioned in this paper is mainly based on the following formula:
\begin{equation}
\label{i1}
A^{-1} = \frac{A^*}{|A|}.
\end{equation}
$|A|$ has already been calculated in the previous section, so this section mainly computes $A^*$. We begin with the following simple theorem:

\begin{proposition}
The inverse of a circulant matrix is still a circulant matrix \cite{horn2012matrix}.
\end{proposition}

Since a circulant matrix can be represented as a vector, we only need to compute one entire row of the inverse $k$-CM matrix, which means $n$ distinct inverse elements.

Following we will introduce our algorithm using $k=5$ as an example, which can solve the inverse of $k$-CM with a complexity of $O(k^3 \log n+k^4)+kn$.

We first perform row and column swaps on the matrix to achieve the following form, requiring only one non-zero element in the upper right corner. Then compute the algebraic cofactors of these $k-1$ elements in boxes. By dividing them by $|M|$, we obtain the first $k-1$ elements of the first column of the inverse matrix.

\begin{equation}
\label{i2}
\setlength{\arraycolsep}{3.5pt}
M = 
\begin{bmatrix}
	\boxed{x_{2}} & \boxed{x_{3}} & \boxed{x_{4}} & \boxed{x_{5}} & 0 & \cdots & 0 & x_{1}\\
	x_{1} & x_{2} & x_{3} & x_{4} & x_{5} & \ddots & 0 & 0 \\
	0 & x_{1} & x_{2} & x_{3} & x_{4} & \ddots & 0 & 0 \\
	\vdots & \ddots & \ddots & \ddots & \ddots & \ddots & \vdots & \vdots \\
	0 &  0 & 0 & \cdots & x_{2} & x_{3} & x_{4} & x_{5}\\ 
	x_{5} & 0 & 0 & \cdots & x_{1} & x_{2} & x_{3} & x_{4} \\
	x_{4} & x_{5} & 0 & \cdots & 0 & x_{1} & x_{2} & x_{3} \\
	x_{3} & x_{4} & x_{5} & \cdots & 0 & 0 & x_{1} & x_{2} \\
\end{bmatrix}
\end{equation}

The algebraic cofactors we need can be obtained using the fast algorithm from Section 2. Their values are equal to the determinant of the following block matrix. We denote the value of the $i$th element we are seeking as $y_i$. $A_i$ and $C_i$ represent the matrices obtained by removing the $i$th column from matrices $A$ and $C$, respectively.

\begin{equation}
\label{i3}
y_i= 
\frac{(-1)^{(i+1)+k(n-2k+2)}x_{k}^{n-2k+2}
	\begin{vmatrix}
		T^{n-2k+2}A_i & B\\
		C_i & D
\end{vmatrix}}{|M|} \ \ (i=1,2,\cdots, k-1)
\end{equation}
\begin{equation}
\setlength{\arraycolsep}{3.5pt}
\nonumber
A=
\begin{bmatrix}
	x_{1} & x_{2} & \cdots & x_{k-1}\\
	0 & x_{1} & \cdots & x_{k-2}\\
	\vdots & \vdots & \ddots & \vdots\\
	0 & 0 & \cdots & x_{1}
\end{bmatrix}_{(k-1)^2},
B=
\begin{bmatrix}
	x_{k} & 0 & \cdots & 0\\
	x_{k-1} & x_{k} & \cdots & 0\\
	\vdots & \vdots & \ddots & \vdots\\
	x_{2} & x_{3} & \cdots & x_{k}
\end{bmatrix}_{(k-1)^2},
T=
\begin{bmatrix}
	-\frac{x_{k-1}}{x_{k}} & 1 & 0  & \cdots & 0\\
	-\frac{x_{k-2}}{x_{k}} & 0 & 1  & \cdots & 0\\
	\vdots & \vdots & \ddots & \ddots & \vdots\\
	-\frac{x_{2}}{x_{k}} & 0 & \cdots & 0 & 1\\
	-\frac{x_{1}}{x_{k}}  & 0 & \cdots & 0 & 0
\end{bmatrix}_{(k-1)^2},
\end{equation}
\begin{equation}
\setlength{\arraycolsep}{3.5pt}
\nonumber
C=
\begin{bmatrix}
	x_{k} & 0 & \cdots & 0 &0\\
	x_{k-1} & x_{k} & \cdots & 0 &0\\
	\vdots & \vdots & \ddots & \vdots &\vdots\\
	x_{3} & x_{4} & \cdots & x_{k} &0 \\
\end{bmatrix}_{(k-2)*{(k-1)}},
D=
\begin{bmatrix}
	x_{1} & x_{2} & \cdots & x_{k-2} & x_{k-1}\\
	0 & x_{1} & \cdots & x_{k-3} & x_{k-2}\\
	\vdots & \vdots & \ddots & \vdots &\vdots \\
	0 & 0 & \cdots & x_{1} & x_{2}
\end{bmatrix}_{(k-2)*{(k-1)}}.
\end{equation}

%A=
%\begin{bmatrix}
%	x_{1} & x_{2} & x_{3} & x_{4} \\
%	0 & x_{1} & x_{2} & x_{3} \\
%	0 & 0 & x_{1} & x_{2}\\
%	0 & 0 & 0 & x_{1} 
%\end{bmatrix},
%B=
%\begin{bmatrix}
%	x_{5} & 0 & 0 & 0\\
%	x_{4} & x_{5} & 0 & 0\\
%	x_{3} & x_{4} & x_{5} & 0\\
%	x_{2} & x_{3} & x_{4} & x_{5}
%\end{bmatrix},
%C=
%\begin{bmatrix}
%	x_{5} & 0 & 0 & 0\\
%	x_{4} & x_{5} & 0 & 0\\
%	x_{3} & x_{4} & x_{5} & 0\\
%\end{bmatrix},
%D=
%\begin{bmatrix}
%	x_{1} & x_{2} & x_{3} & x_{4} \\
%	0 & x_{1} & x_{2} & x_{3} \\
%	0 & 0 & x_{1} & x_{2}\\
%\end{bmatrix}, T =\begin{bmatrix}
%	-\frac{x_{4}}{x_{5}} & 1 & 0 & 0 \\
%	-\frac{x_{3}}{x_{5}} & 0 & 1 & 0 \\
%	-\frac{x_{2}}{x_{5}} & 0 & 0 & 1 \\
%	-\frac{x_{1}}{x_{5}} & 0 & 0 & 0
%\end{bmatrix}

If we set the first column of the inverse matrix as $y$, we can list the following set of equations:
\begin{equation}
\label{i4}
\begin{cases}
	x_1y_n+x_2y_1+x_3y_2+x_4y_3+ \cdots +x_{k}y_{k-1} &= 1\\
	x_1y_1+x_2y_2+x_3y_3+x_4y_4+ \cdots +x_{k}y_{k} & = 0\\ 
	x_1y_2+x_2y_3+x_3y_4+x_4y_5+ \cdots +x_{k}y_{k+1} &= 0\\ 
	x_1y_3+x_2y_4+x_3y_5+x_4y_6+ \cdots +x_{k}y_{k+2} &= 0\\ 
	\ \ \ \ \ \ \ \ \ \ \ \ \ \ \ \ \ \ \ \ \ \ \ \ \ \ \ \ \ \ \  \vdots \\
	x_1y_{n-1}+x_2y_n+x_3y_1+x_4y_2+ \cdots +x_{k}y_{k-2} &= 0\\ 
\end{cases},
\end{equation}
where $x_1$ to $x_{k}$ are known quantities and $y_1$ to $y_{k-1}$ have already been solved in previous steps. Therefore, we can directly derive the other elements in vector $y$ through the following recurrence formula.

\begin{equation}
y_{i} = -\frac{ x_1y_{i-k+1}+x_2y_{i-k+2}+x_3y_{i-k+3}+ \cdots x_{k-1}y_{i-1}}{x_k} \ \ (i=k,k+1,\cdots,n).
\end{equation}
We summarize the above calculation method into Algorithm \ref{a1}.
\begin{algorithm}[h]
\SetAlgoLined
\label{a1}
Input a matrix $M$ in the form of Eq. \ref{i2} \;
calculate $|M|$\;
calculate $y_1,y_2,\cdots,y_{k-1}$ according to Eq. \ref{i3}\;
calculate $y_k,y_{k+1},\cdots,y_{n}$ according to Eq. \ref{i4}\;
construct and return $M^{-1}$
\caption{Inverse of $k$-diagonal circulant matrix $M$.}
\end{algorithm}

\begin{proposition}
We can solve the inverse of a $k$-CM in $O(k^3 \log n+k^4)+kn$.
\end{proposition}
\begin{proof}
We have already proven that the complexity of step 2 is $O(k^3 \log n)$. In step 3, $T^{n-2k+2}$ only needs to be calculated once, with a complexity of $O(k^3 \log n)$. The complexity of other matrix operations is $O(k^4)$. In step 4, the number of calculations for each unknown is $k$, so the cost of step 4 is $kn$.
\end{proof}

We believe that the complexity of this algorithm may have already been reduced to the optimization. In general, $k$ is much smaller than $n$, so the coefficient of $n$ has the greatest impact on the algorithm's speed. Through the method of recurrence relations, we have reduced the coefficient of $n$ to $k$, and even the constant term is only 1. We think that it is almost impossible to continue optimizing further. Furthermore, this algorithm is extremely simple, requiring only five steps, and each step can be directly substituted into a formula, which is also a very commendable aspect.

\section{Inverse of $k$-diagonal cyclic banded matrices\label{s2}}

The calculation of the inverse of $k$-CBM is similar to that of $k$-CM. It still involves first solving for the determinant, then using the algebraic cofactors to solve for some elements in the inverse matrix, and finally solving the entire inverse matrix by solving a system of linear equations. Because the calculation process is mostly similar, we will omit many of the computational steps in this section. However, we will still present the final calculation formulas.

The determinant of a $k$-CBM satisfies the following equation:
\begin{equation}
\begin{vmatrix}M_{n * n}\end{vmatrix} = (-1)^{(k-i)(n-k+i)} 
\begin{vmatrix}D-CA^{-1}(\prod_{j=1}^{n-2k+2}T_j)B\end{vmatrix}  \prod_{j=1}^{n-k+1}x^{(j)}_k,
\end{equation}
where  $n \ge 2k-2$ and the element in the upper left corner of the matrix is $x_{i}$ $(i \leq k$, in example of Eq.\ref{M}, $i = 3$).
\begin{equation}
\renewcommand\arraystretch{1.5}
A = 
\begin{bmatrix}
	x^{(n-2k+3)}_{k} & 0 & \cdots & 0\\
	x^{(n-2k+4)}_{k-1} & x^{(n-2k+4)}_{k}  & \cdots & 0\\
	\vdots & \vdots & \ddots & \vdots\\
	x^{(n-k+1)}_{2} & x^{(n-k+1)}_{3} & \cdots & x^{(n-k+1)}_{k}
\end{bmatrix},
B = 
\begin{bmatrix}
	x^{(1)}_1 & x^{(1)}_2 & \cdots & x^{(1)}_{k-1}\\
	0 & x^{(2)}_1 & \cdots & x^{(2)}_{k-2}\\
	\vdots & \vdots  & \ddots & \vdots\\
	0 & 0 & \cdots & x^{(k-1)}_{1}
\end{bmatrix} ,
T_j =
\begin{bmatrix}
	-\frac{x^{(j+1)}_{k-1}}{x^{(j)}_k} & 1 & 0  & \cdots & 0\\
	-\frac{x^{(j+2)}_{k-2}}{x^{(j)}_k} & 0 & 1  & \cdots & 0\\
	\vdots & \vdots & \ddots & \ddots & \vdots\\
	-\frac{x^{(j+k-2)}_{2}}{x^{(j)}_k} & 0 & \cdots & 0 & 1\\
	-\frac{x^{(j+k-1)}_{1}}{x^{(j)}_k}  & 0 & \cdots & 0 & 0
\end{bmatrix},
\end{equation}
\begin{equation}
\label{e11}
\renewcommand\arraystretch{1.5}
C = 
\begin{bmatrix}
	x^{(n-k+2)}_{1} & x^{(n-k+2)}_{2} & \cdots & x^{(n-k+2)}_{k-1}\\
	0 & x^{(n-k+3)}_{1} & \cdots & x^{(n-k+3)}_{k-2}\\
	\vdots & \vdots & \ddots & \vdots \\
	0 & 0 & \cdots & x^{(n)}_{1}
\end{bmatrix},
D = 
\begin{bmatrix}
	x^{(n-k+2)}_{k} & 0 & \cdots & 0\\
	x^{(n-k+3)}_{k-1} & x^{(n-k+3)}_{k} & \cdots & 0\\
	\vdots & \vdots & \ddots & \vdots\\
	x^{(n)}_{2} & x^{(n)}_{3} & \cdots & x^{(n)}_{k}
\end{bmatrix}.
\end{equation}

We can quickly conclude that the complexity of calculating the determinant of a $k$-CBM is $O(k^3n)$.
Next, unlike with $k$-CM, we need to calculate the algebraic cofactors of the $(k-1)^2$ elements marked in boxes. For ease of demonstration, we will still use $k=5$ as an example.

\begin{equation}
\label{CBM}
M = 
\begin{bmatrix}
	\boxed{x_{2}^{(1)}} & \boxed{x_{3}^{(1)}} & \boxed{x_{4}^{(1)}} & \boxed{x_{5}^{(1)}} & 0 & 0 & 0 & \cdots & 0 & x_{1}^{(1)}\\
	x_{1}^{(2)} & \boxed{x_{2}^{(2)}} & \boxed{x_{3}^{(2)}} & \boxed{x_{4}^{(2)}} & \boxed{x_{5}^{(2)}} & 0 & 0 & \ddots & 0 & 0 \\
	0 & x_{1}^{(3)} & \boxed{x_{2}^{(3)}} & \boxed{x_{3}^{(3)}} & \boxed{x_{4}^{(3)}} & \boxed{x_{5}^{(3)}} & 0 & \ddots & 0 & 0 \\
	0 & 0 & x_{1}^{(4)} & \boxed{x_{2}^{(4)}} & \boxed{x_{3}^{(4)}} & \boxed{x_{4}^{(4)}} & \boxed{x_{5}^{(4)}} & \ddots & 0 & 0  \\
	\vdots & \ddots & \ddots & \ddots & \ddots & \ddots & \ddots & \ddots & \vdots & \vdots \\
	0 &  0 & 0 & \cdots & 0 & x_{1}^{(n-3)} & x_{2}^{(n-3)} & x_{3}^{(n-3)} & x_{4}^{(n-3)} & x_{5}^{(n-3)}\\ 
	x_{5}^{(n-2)} & 0 & 0 & \cdots & 0 & 0 & x_{1}^{(n-2)} & x_{2}^{(n-2)} & x_{3}^{(n-2)} & x_{4}^{(n-2)} \\
	x_{4}^{(n-1)} & x_{5}^{(n-1)} & 0 & \cdots & 0 & 0 & 0 & x_{1}^{(n-1)} & x_{2}^{(n-1)} & x_{3}^{(n-1)} \\
	x_{3}^{(n)} & x_{4}^{(n)} & x_{5}^{(n)} & \cdots & 0 & 0 & 0 & 0 & x_{1}^{(n)} & x_{2}^{(n)} \\
\end{bmatrix}
\end{equation}

For the algebraic cofactors of the elements in the first row, we have already calculated them previously. For the elements in the second row, we can calculate them in the same manner by moving the first row of the matrix to the bottom and the first column to the rightmost, and so on. By dividing the results by $|M|$, we obtain the $(k-1)^2$ elements in the inverse matrix. We use $M^{-1}_{i,j}$ to represent the element in the $i$th row and $j$th column of the matrix $M^{-1}$. $A_{i,j}$ and $C_{i,j}$ represent the matrices obtained by removing the $i$th column from matrices $A_{j}$ and $C_{j}$, respectively. Moreover, when the indices (superscripts and subscripts $i$ and $j$) exceed $n$, they should automatically modulo $n$ and similarly, if they are less than $n$, then $n$ should automatically be added.
\begin{equation}
\label{CBM-1}
M^{-1}_{i+j-1,j}= 
\frac{(-1)^{(i+1)+k(n-2k+2)}\prod_{p=j+1}^{n-2k+2+j}x^{(p)}_k
	\begin{vmatrix}
		\prod_{p=j+1}^{n-2k+2+j}T_pA_{i,j} & B_{j}\\
		C_{i,j} & D_{j}
\end{vmatrix}}{|M|} \ \ (i, j = 1,2,\cdots ,k-1)
\end{equation}
\begin{equation}
\setlength{\arraycolsep}{2pt}
\nonumber
A_{j}=
\begin{bmatrix}
	x_{1}^{(j+1)} & x_{2}^{(j+1)} & \cdots & x_{k-1}^{(j+1)}\\
	0 & x_{1}^{(j+2)} & \cdots & x_{k-2}^{(j+2)}\\
	\vdots & \vdots & \ddots & \vdots\\
	0 & 0 & \cdots & x_{1}^{(j+k-1)}
\end{bmatrix}_{(k-1)^2},
B_{j}=
\begin{bmatrix}
	x_{k}^{(n+j-2k+3)} & 0 & \cdots & 0\\
	x_{k-1}^{(n+j-2k+4)} & x_{k}^{(n+j-2k+4)} & \cdots & 0\\
	\vdots & \vdots & \ddots & \vdots\\
	x_{2}^{(n+j-k+1)} & x_{3}^{(n+j-k+1)} & \cdots & x_{k}^{(n+j-k+1)}
\end{bmatrix}_{(k-1)^2},
T_{j} =
\begin{bmatrix}
	-\frac{x^{(j+1)}_{k-1}}{x^{(j)}_k} & 1 & 0  & \cdots & 0\\
	-\frac{x^{(j+2)}_{k-2}}{x^{(j)}_k} & 0 & 1  & \cdots & 0\\
	\vdots & \vdots & \ddots & \ddots & \vdots\\
	-\frac{x^{(j+k-2)}_{2}}{x^{(j)}_k} & 0 & \cdots & 0 & 1\\
	-\frac{x^{(j+k-1)}_{1}}{x^{(j)}_k}  & 0 & \cdots & 0 & 0
\end{bmatrix}_{(k-1)^2}
\end{equation}
\begin{equation}
\setlength{\arraycolsep}{2pt}
\nonumber
C_{j}=
\begin{bmatrix}
	x_{k}^{(n+j-k+2)} & 0 & \cdots & 0 &0\\
	x_{k-1}^{(n+j-k+3)} & x_{k}^{(n+j-k+3)} & \cdots & 0 &0\\
	\vdots & \vdots & \ddots & \vdots &\vdots\\
	x_{3}^{(n+j-1)} & x_{4}^{(n+j-1)} & \cdots & x_{k}^{(n+j-1)} &0 \\
\end{bmatrix}_{(k-2)*{(k-1)}},
D_{j}=
\begin{bmatrix}
	x_{1}^{(n+j-k+2)} & x_{2}^{(n+j-k+2)} & \cdots & x_{k-2}^{(n+j-k+2)} & x_{k-1}^{(n+j-k+2)}\\
	0 & x_{1}^{(n+j-k+3)} & \cdots & x_{k-3}^{(n+j-k+3)} & x_{k-2}^{(n+j-k+3)}\\
	\vdots & \vdots & \ddots & \vdots &\vdots \\
	0 & 0 & \cdots & x_{1}^{(n+j-1)} & x_{2}^{(n+j-1)}
\end{bmatrix}_{(k-2)*{(k-1)}}.
\end{equation}

Now that we have $k-1$ known elements in one column, we can set up a system of equations to solve for all the elements in these $k-1$ columns. For each $i=1,2,\cdots ,n$ and $j=1,2,\cdots ,k-1$, we have

\begin{equation}
\label{CBM-2}
M^{-1}_{i,j}=-\frac{ x_{1}^{(i-k+2)}M^{-1}_{i-k+1,j}+x_{2}^{(i-k+2)}M^{-1}_{i-k+2,j}+x_{3}^{(i-k+2)}M^{-1}_{i-k+3,j}+ \cdots + x_{k-1}^{(i-k+2)}M^{-1}_{i-1,j}}{x_{k}^{(i-k+2)}} .
\end{equation}

If $AB=E$, then $BA=E$ as well. We can use these known $k-1$ columns to continue setting up equations, and thereby quickly solve for the entire matrix. For each $i=1,2,\cdots , n$ and $j = k, k+1, \cdots ,n$, we have

\begin{equation}
\label{CBM-3}
M^{-1}_{i,j}=
\begin{cases}
	-\frac{x_{k}^{(j-k+1)}M^{-1}_{i,j-k+1}+x_{k-1}^{(j-k+2)}M^{-1}_{i,j-k+2}+x_{k-2}^{(j-k+3)}M^{-1}_{i,j-k+3}+ \cdots +x_{2}^{(j-1)}M^{-1}_{i,j-1}}{x_{1}^{(j)}} & \text{if}\ i \neq j-1 \\
	\frac{1-x_{k}^{(j-k+1)}M^{-1}_{i,j-k+1}-x_{k-1}^{(j-k+2)}M^{-1}_{i,j-k+2}-x_{k-2}^{(j-k+3)}M^{-1}_{i,j-k+3}- \cdots -x_{2}^{(j-1)}M^{-1}_{i,j-1}}{x_{1}^{(j)}} & \text{if}\ i = j-1
\end{cases}.
\end{equation}

We summarize the above calculation method into Algorithm \ref{a2}.
\begin{algorithm}[h]
\SetAlgoLined
\label{a2}
Input a matrix $M$ in the form of Eq. \ref{CBM} \;
calculate $|M|$\;

\For{$i = 1; i < k; i++$}{
	\For{$j = 1; j < k; j++$}{
		calculate $M^{-1}_{i+j-1,j}$ according to Eq. \ref{CBM-1}\;
	}
}
\For{$j = 1; j < k; j++$}{
	\For{$i = k-1+j; i < n+j; i++$}{
		calculate $M^{-1}_{((i-1)\ mod\ n)+1,j}$ according to Eq. \ref{CBM-2}\;
	}	
}
\For{$i=1; i \leq n; i++$}{
	\For{$j = k;j \leq n; j++$}{
		calculate $M^{-1}_{i,j}$ according to Eq. \ref{CBM-3}\;
	}
	
}
construct and return $M^{-1}$
\caption{Inverse of $k$-diagonal cyclic banded matrix $M$.}
\end{algorithm}

\begin{proposition}
We can solve the determinant of a $k$-CBM in $O(k^3 n+k^5)+kn^2$.
\end{proposition}
\begin{proof}
In steps 3 to 7, $\prod_{p=k}^{n-2k+3}T_p$ and $\prod_{p=k}^{n-2k+3}x_{k}^{(p)}$ only needs to be calculated once, with a complexity of $O(k^3 n)$. The complexity of other matrix operations is $O(k^5)$. In steps 8 to 17, the number of calculations for each unknown is $k$, so its cost is $kn^2$.
\end{proof}

In cases where $k$ is much smaller than $n$, the inverse of a $k$-CBM can almost be found within a complexity of $kn^2$ through an iterative solution of linear equations. Notably, the coefficient for $kn^2$ is 1. Moreover, this algorithm is also very simple and easy to understand. This makes the algorithm exceptionally excellent from any perspective.

\section{Conclusion \label{s3}}
In this paper, we first introduced a method for calculating the determinant of $k$-CMs with a complexity of $O(k^3\log n)$. Based on this, we presented fast computation methods for the inverses of $k$-CMs and $k$-CBMs, with complexities of $O(k^3 \log n+k^4)+kn$ and $O(k^3 n+k^5)+kn^2$, respectively. These two calculation methods have the following advantages. First, their inverses can be rapidly represented. Second, most steps are directly substituted into formulas. Finally, the coefficients of $kn$ and $kn^2$ have been reduced to the minimum. Finally, our computations are also applicable to finite fields.

\section*{Acknowledgments}
This study is financed by Tianjin Science and Technology Bureau, finance code: 21JCYBJC00210.


%% References


\bibliographystyle{elsarticle-harv}
\bibliography{fbref}






% \end{thebibliography}

\end{document}